\documentclass[aip,jcp,reprint,amsmath,amsfonts,amssymb,a4paper,eqsecnum,superscriptaddress]{revtex4-1}

\usepackage[usenames,dvipsnames]{color}
\usepackage{color,graphicx,amsthm}
\usepackage[utf8]{inputenc}
\usepackage{natbib}
\newtheorem{definition}{Definition}

\newtheorem{theorem}{Theorem}
\newtheorem{example}{Example}

\usepackage[T1]{fontenc}
\newcommand\R{\mathbb{R}}

\begin{document}

\title{Ehrenfest dynamics is purity non-preserving:  a necessary ingredient for 
decoherence}

\author{J. L. Alonso}
\affiliation{Departamento de F{\'{\i}}sica Te\'orica, Universidad de Zaragoza, 
Pedro Cerbuna 12, E-50009 Zaragoza, Spain}
\affiliation{Instituto de Biocomputaci\'on y F{\'{\i}}sica de Sistemas 
Complejos (BIFI), Universidad de Zaragoza, Mariano Esquillor s/n, Edificio 
I+D, E-50018 Zaragoza, Spain}
\affiliation{Unidad Asociada IQFR-BIFI,Universidad de Zaragoza, Mariano Esquillor s/n, Edificio 
I+D, E-50018 Zaragoza, Spain}
\author{J. Clemente-Gallardo}
  \affiliation{Departamento de F{\'{\i}}sica Te\'orica, Universidad de Zaragoza, 
Pedro Cerbuna 12, E-50009 Zaragoza, Spain}
\affiliation{Instituto de Biocomputaci\'on y F{\'{\i}}sica de Sistemas 
Complejos  (BIFI), Universidad de Zaragoza, Mariano Esquillor s/n, Edificio 
I+D, E-50018 Zaragoza, Spain}
\affiliation{Unidad Asociada IQFR-BIFI,Universidad de Zaragoza, Mariano Esquillor s/n, Edificio 
I+D, E-50018 Zaragoza, Spain}
\affiliation{Zaragoza Scientific Center for Advanced Modeling (ZCAM), Universidad de Zaragoza, Mariano Esquillor s/n, Edificio I+D, E-50018 Zaragoza, Spain}
\author{J. C. Cuch\'{\i}}
\affiliation{Departament d'Enginyeria Agroforestal, ETSEA-Universitat de 
Lleida, Av. Alcalde Rovira Roure 191, 25198 Lleida, Spain}
\author{P. Echenique}
\affiliation{Instituto de Qu\'{\i}mica F\'{\i}sica Rocasolano, CSIC, Serrano 
119, E-28006 Madrid, Spain}
\affiliation{Departamento de F{\'{\i}}sica Te\'orica, Universidad de Zaragoza, 
Pedro Cerbuna 12, E-50009 Zaragoza, Spain}
\affiliation{Instituto de Biocomputaci\'on y F{\'{\i}}sica de Sistemas 
Complejos (BIFI), Universidad de Zaragoza, Mariano Esquillor s/n, Edificio
I+D, E-50018 Zaragoza, Spain}
\affiliation{Zaragoza Scientific Center for Advanced Modeling (ZCAM), Universidad de Zaragoza, Mariano Esquillor s/n, Edificio I+D, E-50018 Zaragoza, Spain}
\affiliation{Unidad Asociada IQFR-BIFI,Universidad de Zaragoza, Mariano Esquillor s/n, Edificio 
I+D, E-50018 Zaragoza, Spain}
\author{F. Falceto}
\affiliation{Departamento de F{\'{\i}}sica Te\'orica, Universidad de Zaragoza, 
Pedro Cerbuna 12, E-50009 Zaragoza, Spain}
\affiliation{Instituto de Biocomputaci\'on y F{\'{\i}}sica de Sistemas 
Complejos (BIFI), Universidad de Zaragoza, Mariano Esquillor s/n, Edificio 
I+D, E-50018 Zaragoza, Spain}

\begin{abstract}
We discuss the evolution of purity in mixed quantum/classical approaches to
electronic nonadiabatic dynamics in the context of the Ehrenfest model.
 As it is impossible to exactly determine initial conditions for a
realistic system, we choose to work in the statistical Ehrenfest formalism
that we introduced in Ref.~\onlinecite{Alonso2011JPA}. From it, we develop a
new framework to determine exactly the change in the purity of the quantum
subsystem along the evolution of a statistical Ehrenfest system. In a simple
case, we verify how and to which extent Ehrenfest statistical dynamics makes a
system with more than one classical trajectory and an initial quantum pure
state become a quantum mixed one. We prove this numerically
showing how the evolution of purity depends on time, on the dimension of the
quantum state space $D$, and on the number of classical trajectories $N$ of
the initial distribution. The results in this work open new perspectives for
studying decoherence with Ehrenfest dynamics.
\end{abstract}
\maketitle

\section{Introduction}

The Schr\"odinger equation for a combined system of electrons and nuclei
enables us to predict most of the chemistry and molecular physics that
surrounds us, including biophysical processes of great complexity.
Unfortunately, this task is not possible in general, and approximations need
to be made; one of the most important and successful being the classical
approximation for a number of the particles. Mixed quantum-classical dynamical
(MQCD) models are therefore necessary and widely used.

We could say that, typically, the technique used to build MQCD models is a
partial `deconstruction' of the quantum mechanics (QM) of the total system
(electrons and nuclei) followed by a `reconstruction' that tries to recover
the essential properties of the total Schr\"odinger equation lost in the
deconstruction process. It is unrealistic to expect the reconstructed theory
has the same predictive power as the Schr\"odinger equation, so the
reconstructed theory will apply with enough accuracy only to a subset of
systems and questions; a subset whose boundaries are difficult to predict a
priori. In the literature, there are at least two common levels of
deconstruction, one further away from the total Schr\"odinger equation for
electrons and nuclei, called Born-Oppenheimer molecular dynamics (BOMD), where
electrons are assumed to remain in the ground state for all times, and another
one closer to it, called Ehrenfest dynamics (ED), where nuclei are still
classical (as in BOMD) but the electrons are allowed to populate excited
states (some misleading notation used in the literature on ED is
clarified in Section 2 of Ref. \onlinecite{Andrade2009JCTC}).

In J. C. Tully's surface hopping methods,\cite{Tully1990JCP} for example, the
deconstruction goes to BOMD and the reconstruction proceeds by allowing the
system to perform certain specially designed stochastic jumps between
adiabatic states.

In the decay of mixing formalism by D. G. Truhlar and
coworkers,\cite{Zhu2005JCTC} the deconstruction stops at the ED and the
reconstruction is developed by adding decoherence to it. This has been shown
to be more accurate than surface hopping methods for non-Born-Oppenheimer
collisions.

 When considered for a single system, ED  is a fully
coherent semiclassical method, and hence purity preserving. As decoherence
must be a property of any realistic
model,\cite{Tully1990JCP,Zhu2005JCTC,Truhlar2007Book} many MQCD models have
been reconstructed in order to produce electronic decoherence. See, for
example, Refs.~\color{black}\onlinecite{Truhlar2007Book,
Zhu2005JCTC, Prezhdo1999JCP, Subotnik2010JCP, Horsfield2006RPP, Tully1990JCP,
Tully1990b, Tully1998Book,
Landry2011JCP,Bedard2005JCP,SunWangMiller:1998,SubShen:2011}, that range from
one of the most classic in this matter \cite{Tully1990JCP} to one of the most
recent.\cite{SubShen:2011}. In Ref.~\onlinecite{SubShen:2011}, we
can find a recent study of decoherence in the context of surface hopping and
an important conclusion: averaging over a swarm of initial conditions,
decoherence can be measured; but the method cannot capture all the observed
effects (for example, the averaging is not enough to capture the physics of
wave packet bifurcations on multiple surfaces). In the case of decay of mixing
formalisms, where the trajectories in the swarm are considered as independent,
decoherence phenomena are incorporated algorithmically (see Eq.~18 of
Ref.~\onlinecite{Truhlar2007Book}).

 In this paper we study the problem from a different perspective.
First of all, we consider a complete statistical description of ED, which was
introduced in Ref.~\onlinecite{Alonso2011JPA} in full detail. Based on that
construction, we develop a description of Quantum Statistical Mechanics which
can be adapted to mixed quantum-classical systems in a straightforward but
rigorous manner. This allows us to consider, in a simple way, the evolution of
the purity of the quantum subsystem of our Ehrenfest model. We prove thus
that, while a single Ehrenfest system evolves preserving the purity of the
quantum state, the behavior changes dramatically when a statistical
distribution is considered and, in general, it introduces a change in the
purity of the quantum system along the trajectory. Therefore, we can claim
that our general statistical ED formalism is purity non-preserving; a property which
always accompanies decoherence phenomena (see for example Sec.~3.5 in
Ref.~\onlinecite{Schlosshauer2007}). This work opens thus a new line of
research in the direction of taking into account decoherence phenomena
in ED.
Of course, those ingredients  still missing for a proper description of
decoherence in our Ehrenfest statistical formalism could be later added in a
reconstruction process similar to the ones mentioned before, but this time
starting from a, presumably better, purity non-preserving
dynamics

A full study of the decoherence process  is a very complex task which includes deep quantum
theoretic concepts as the measurement problem and the interpretations
of QM
 (see, for example, Refs.~\onlinecite{Zurek2003RMP,
Schlosshauer2007} for a general discussion, and Ref.~\onlinecite{Zhu2005JCTC}
for the analysis of the decoherence phenomenon on molecular systems).
 In this paper we voluntarily restrict ourselves to a simple
property reflecting the decoherence phenomenon. Namely, the change in the
`degree of mixture' of the quantum state in a MQCD model, as quantified by the
purity $\mathrm{Tr}\rho^2$. As mentioned, we shall see in Sections
\ref{sec:purity-non-pres} and \ref{sec:ejemplo} 
that ED provides a framework where this change takes place. The actual
relation with the electronic decoherence in molecular systems requires a much
more involved analysis which will be developed in the future.

Besides the incomplete description of decoherence, usual
approaches to ED have  also been often criticized on the basis
that it does not yield the Boltzmann equilibrium distribution for the
electrons
exactly.\cite{Mauri1993EL,Parandekar2006JCTC,Parandekar2005JCP,Schmidt2008JCP,
Bastida2007JCP} The lack of this property, which we agree is desirable, is
however not enough to rule out ED for all applications, as we recently argued.
\cite{Alonso2010NJP, Alonso2012}

The structure of the paper is as follows: Sections
from~\ref{sec:notion-decoherence} to~\ref{sec:purity-non-pres} introduce the
mathematical formalism and the relevant definitions, which are then put into
practice in the numerical example in Sec.~\ref{sec:ejemplo}.
Sec.~\ref{sec:notion-decoherence} reviews the notion of purity in QM and
proves the well known fact that ED preserves the purity of the quantum
subsystem when we consider the evolution of a single trajectory from perfectly
determined initial conditions. Sec.~\ref{sec:geom-quant-stat} presents a very
brief summary of the formulation of geometric QM (see
\onlinecite{Alonso2011JPA} for a more careful presentation) and it provides an
analogous formulation of a quantum statistical system within the same
framework. In particular, a suitable formulation of the purity of a quantum
system is introduced. Sec.~\ref{sec:purity-non-pres} presents the main
contribution of the paper: first, we review the geometrical formulation of ED
and its associated statistical equilibrium introduced in
\onlinecite{Alonso2011JPA}. Then, we adapt the tools introduced in the
previous section in order to be able to study the evolution of the purity of
the quantum subsystem in a suitable way, and to show that ED is purity
non-preserving. The use of the geometrical formalism, as we will see in what
follows, allows to perform a very direct analysis of the problem. In
Sec.~\ref{sec:ejemplo}, we numerically illustrate the change in purity
produced by ED using a very simple but extremely useful example: a statistical
system defined by a pure quantum state and an ensemble of initial conditions
of the classical subsystem. Such a system has been used in the literature as a
natural framework for molecular dynamics (see for example
Refs.~\onlinecite{Tully1990b,Tully1998Book,Bastida2007JCP}). We use it as the
simplest nontrivial Ehrenfest statistical system where we can show how the
purity of the quantum part of the system evolves in time depending on the
coupling between the classical and quantum systems, the initial momentum of
the classical particles, the dimension of the quantum state space and the
number of trajectories considered in the initial conditions. In
Sec.~\ref{sec:concl-future-work} we present our conclusions and our plans for
future works.

\section{Purity}
\label{sec:notion-decoherence}

\subsection{Purity preservation in quantum mechanics}
\label{sec:noti-purity-quant}

Given a Hilbert space $\mathcal{H}$, we shall call {\bf density states} to the
elements $\rho$ obtained as convex combinations of rank-one projectors
$\rho=\{ \rho_1, \ldots, \rho_k\}$, each element satisfying
$$
  \rho_j^2=\rho_j,\quad \rho_j^+=\rho_j,   \quad \mathrm{Tr}\rho_j=1, 
  \quad j=1,\ldots,k,
$$ 
with a probability vector, $p := (p_1, \ldots, p_k)$ with
$\sum_jp_j=1$ and $p_j\geq 0$, $\forall j$. The expression of a general
density state is then
$$ 
  \rho=\sum_jp_j\rho_j.  
$$ 
The evaluation of some observable $A$ on this state is given by
\begin{equation}   
\langle A\rangle=\sum_jp_j \mathrm{Tr}(\rho_j A)=\mathrm{Tr}(\rho A).  
 \label{eq:aver_mix}
\end{equation}

The state of the quantum system is said to be \textbf{pure} if the density
matrix which represents it is a rank-one projector, i.e., if the convex
combination above contains only one term. If this property does not hold, the
system is said to be in a \textbf{mixed} state, since, from the physical point
of view, there is a statistical mixture of the different pure states
represented by the density matrices $\rho_j$ above.

Being a Hermitian operator, the matrix $\rho$ can be diagonalized. Its
eigenvalues $\{\lambda_1,\ldots,\lambda_k\}$ satisfy
\begin{equation}
  \label{eq:8}
  0\leq \lambda_{j}\leq 1, \quad \forall j.
\end{equation}
If the state is pure, there is one eigenvalue equal to one, the rest being
zero. Obviously, the rank of $\rho$ as a projector on the Hilbert space
$\mathcal{H}$ coincides with the number of nonvanishing eigenvalues.

It is an immediate property that a state $\rho$ is pure if and only if 
\begin{equation}
   \label{eq:6i}
   \mathrm{Tr}\rho^2=1.
\end{equation}
The proof requires only Eq.~\eqref{eq:8}.

The description of a quantum system in terms of a density matrix uses von
Neumann's equation to introduce the dynamics. Then we know that, given the
Hamiltonian operator $H$, the evolution of the state $\rho$ is given by
\begin{equation}
   \label{eq:7}
   i\hbar \dot \rho(t)=[H(t),\rho(t)],
\end{equation}
where $[\cdot,\cdot]$ is the usual commutator of operators.

Using a simple proof which is formally identical to the one that we shall
present in the next section, one can easily show that this dynamics is
purity-preserving, i.e., $d (\mathrm{Tr}\rho^2 ) / dt = 0$.

\subsection{Purity preservation in non-statistical Ehrenfest dynamics}

The Ehrenfest equations \cite{Alonso2010NJP, Alonso2012} for a system composed
of a set of $M$ classical particles (typically nuclei; described by the phase
space variables $R := (\vec{R}_1,\ldots,\vec{R}_M)$, $P :=
(\vec{P}_1,\ldots,\vec{P}_M)$) and a set of $n$ quantum particles (typically
electrons; described by a wavefunction $\psi$, defined on the space
parameterized by $r := (\vec{r}_1,\ldots,\vec{r}_n)$) are:
{\allowdisplaybreaks
\begin{eqnarray}
  \dot{\vec{R}}_J(t) & = & \frac{\vec{P}_J}{M_J},
  \label{eq:ehrenfest1}
 \\
  \dot{\vec{P}}_J(t) & = & -\langle \psi(t) \vert \frac{\partial
  H_e}{\partial \vec R_J}(R(t)) \vert \psi(t)\rangle,
  \label{eq:ehrenfest2}
 \\
  i\hbar \frac{{\rm d}}{{\rm d}t} |\psi(t)\rangle & = & 
   H_e(R(t))\vert\psi(t)\rangle,
  \label{eq:ehrenfest3}
\end{eqnarray}
}
where $J=1,\ldots,M$ and the electronic Hamiltonian operator $H_e$ is
related to the molecular one $H$ and it is defined as follows:
\begin{align}
\label{eq:29}
H_e(R) :=& - \hbar^2\sum_j\frac 12
 \nabla_j^2+\frac 1{4\pi \epsilon_0}\sum_{J<K}\frac{Z_JZ_K}{|\vec R_J-  \vec
   R_K|} \nonumber \\
&+\frac 1{4\pi \epsilon_0}\sum_{j<k}\frac 1{|\vec r_j-\vec r_k|}-
  \frac 1{4\pi \epsilon_0}\sum_{J,j}\frac {Z_J}{|\vec R_J-\vec
    r_j|}\nonumber \\
= &H + \hbar^2\sum_J\frac 1{2M_J}\nabla_J^2,
\end{align}
where all sums must be understood as running over the whole natural set for
each index, $M_J$ is the mass of the $J$-th nucleus in units of the electron
mass, and $Z_J$ is the charge of the $J$-th nucleus in units of (minus) the
electron charge.
 
At first sight, given the similarity between Eq. \eqref{eq:ehrenfest3} and the
Schrödinger equation for an isolated full-quantum system, one might
erroneously think that the Ehrenfest evolution for the quantum part of the
system is unitary. \cite{Marx:2000,Teilhaber:1992,Kalia1990} If this was
correct, then it would be trivial to prove that ED is purity preserving, but
this is not the case. As is well known, for a one to one transformation, we
can define unitarity as the property of preserving the scalar product, i.e.,
given two arbitrary vectors $\varphi$ and $\phi$, we say that $U$ is unitary
if $\langle U\varphi| U\phi\rangle = \langle\varphi|\phi\rangle$. One can
easily see that any reversible transformation $U$ that enjoys this property is
necessarily linear:
\begin{align*}
\langle &U(\varphi_1+\varphi_2) |U\phi\rangle=
\langle \varphi_1+\varphi_2| \phi\rangle=
\langle\varphi_1| \phi\rangle+\langle \varphi_2| \phi\rangle
\\
&=\langle U\varphi_1| U\phi\rangle+\langle U\varphi_2| U\phi\rangle
=\langle (U\varphi_1+U\varphi_2)| U\phi\rangle.
\end{align*}
As $U$ is reversible, $U\phi$ is an arbitrary vector, and therefore,
we must have,
$$
U(\varphi_1+\varphi_2)=U\varphi_1+U\varphi_2.
$$
But, although the quantum part of the the equations of motion in
\eqref{eq:ehrenfest3} resembles a typical Schr\"odinger equation, the coupling
with the classical part makes the evolution of the quantum system nonlinear.
Consequently, it cannot be a unitary transformation as defined above.

Despite this non-unitarity, it is very simple to prove that, if we consider
the evolution of a single trajectory $(R,P,\psi)$ of an Ehrenfest system, the quantum part is always in a pure state:
\begin{theorem}
  Let $(R,P,\psi)$ be the initial state of an Ehrenfest
  system subject to the dynamics given by eqs. 
  \eqref{eq:ehrenfest1}--\eqref{eq:ehrenfest3}. 
  Then, the quantum part of the system is always in a pure state
\end{theorem}

\begin{proof}
We consider the density matrix $\rho=|\psi\rangle\langle\psi|$ corresponding
to the quantum part of the Ehrenfest system. The evolution of $|\psi\rangle$
is given by Eq.~\eqref{eq:ehrenfest3}, which induces a von Neumann-like
evolution for the density matrix at every time $t$
$$
i\hbar \dot \rho(t)=[H_e(R(t)), \rho(t)],
$$
being $H_e$ the electronic Hamiltonian in~\eqref{eq:29}. Then, 
\begin{align*}
\frac {d}{dt}\mathrm{Tr}\rho^{2}&=2 \mathrm{Tr}(\dot\rho\rho)=
2\mathrm{Tr}([H_e,\rho]\rho)  \\
&=2\left ( \mathrm{Tr}(H_e\rho\rho) -
 \mathrm{Tr}(\rho H_e \rho) \right ) = 0,
\end{align*}
for all times $t$. Hence, if $\mathrm{Tr}\rho^2=1$ at $t=0$, it will remain
so.
\end{proof}

The main goal of the rest of the paper is to prove that, when we consider the
case of a statistical ensemble of Ehrenfest trajectories, this is no longer
the case: the evolution of an ensemble whose quantum part is a pure state at
$t=0$ will become an ensemble in which the quantum part is mixed as long as
the initial conditions for the nuclei are not perfectly determined. Thus, in
such a case, the ED produces a purity change, a necessary condition for
decoherence.

\section{Geometric Quantum Statistical Mechanics}
\label{sec:geom-quant-stat}

\subsection{Geometric quantum mechanics}
\label{sec:append-geom-quant}

The aim of this section is to provide a very brief summary of the mathematical
formalism more thoroughly introduced in \onlinecite{Alonso2011JPA} and
references therein.

Classical mechanics can be formulated in several mathematical frameworks each
corresponding to a different level of abstraction: Newton's equations, the
Hamiltonian formalism, the Poisson brackets, etc. Perhaps its more abstract
and general formulation is geometrical, in terms of Poisson manifolds.
Similarly, QM can also be formulated in different ways, some of which resemble
its classical counterpart  (see Ref.~\onlinecite{MeyMiller1979}
for a classical reference in the context of molecular systems,
Ref.~\onlinecite{Alonso2011JPA} for a more recent one, and
Refs.~\onlinecite{Kibble:1979, Heslot:1985, brody:2001, carinena:2007} for
more mathematical approaches). For example, the observables
(self-adjoint linear operators) are endowed with a Poisson algebra structure
(based on the commutators) almost equal to the one that characterizes the
dynamical variables in classical mechanics. Moreover, Schr\"odinger equation
can be recast into Hamilton's equations form by transforming the complex
Hilbert space into a real one of double dimension. The observables are also
transformed into dynamical functions in this new phase space, in analogy to
the classical one. Finally, a Poisson bracket formulation has also been
established for QM, which permits to classify both the classical and the
quantum dynamics under the same heading.

This variety of formulations does not emerge from academic caprice; the
succesive abstractions simplify further developments of the theory, such as
the step from microscopic dynamics to statistical dynamics: the derivation of
Liouville's equation (or von Neumann's equation in the quantum case), at the
heart of statistical dynamics, is based on the properties of the Poisson
algebra.\cite{Alonso2011JPA}

Consider a basis $\{ |\psi_{k}\rangle\}$ for the Hilbert space ${\cal H}$.
Each state $|\psi\rangle\in {\cal H}$ can be written in that basis with
complex components (or coordinates, in more differential geometric terms) $\{
z_{k}\}$:
$$
{\cal H}\ni |\psi\rangle=\sum_{k}z_{k}|\psi_{k}\rangle.
$$

Now, we can just take the original vector space inherent to the Hilbert space,
and turn it into a real vector space (denoted as $M_Q$), by splitting each
coordinate into its real and imaginary parts:
$$
\mathbb{C}^n\sim \mathcal{H}\ni z_{k}=q^{k}+ip_{k} \mapsto (q^{k}, p_{k})\in 
\mathbb{R}^{2n}\equiv M_{Q}.
$$
We will use real coordinates $(q^k, p_k)$, $k=1, \ldots,n$, to represent the
points of $\mathcal{H}$ when thought of as real vector space elements. From
this point of view, the similarities between the quantum dynamics and the
classical one will be more evident. It is important to notice, though, that
despite the formal similarities these coordinates $(q^k, p_k)$ do not
represent physical positions and momenta of any actual system. They simply
correspond to the real and imaginary parts of the complex coordinates used for
the Hilbert space vectors in a given basis.

The scalar product of the Hilbert space is encoded in three tensors defined on
the real vector space $M_Q$. The interested reader is addressed to
Ref.~\onlinecite{Alonso2011JPA} for the details. We just highlight here that
two of these tensors correspond to a metric tensor $g$ and a symplectic one
$\omega$ which allow us to write the expression of the Schr\"odinger equation
as a Hamilton equation, in a form which is completely analogous to the
Hamiltonian formulation of classical mechanics. It is precisely this
similarity the key ingredient to successfully combine classical and quantum
mechanics in a well-defined framework to describe the Ehrenfest equations
\eqref{eq:ehrenfest1}--\eqref{eq:ehrenfest3} as a Hamiltonian system, as we
will summarize later and it can also be seen in
Ref.~\onlinecite{Alonso2011JPA}.

In this formalism, instead of considering the observables as linear operators
(plus the usual requirements, self-adjointness, boundedness, etc.) on the
Hilbert space $\mathcal{H}$, we shall represent them as functions defined on
the real space $M_Q$. The reason for that is to resemble, as much as possible,
the classical mechanical approach. But we cannot forget the linearity of the
operators, and thus the functions must be chosen in a very particular way. The
usual choice is inspired in Ehrenfest's description of quantum mechanical
systems and defines, associated to any operator $A$ on $\mathcal{H}$, a
function of the form:
\begin{equation}
\label{eq:2}
f_A(\psi):=\frac 12\langle\psi,A\psi\rangle.
\end{equation}

The operations which are defined on the set of operators can also be
translated into this new language. Thus, the associative product of operators
(the matrix product when considered in a finite dimensional Hilbert space),
the commutator (which encodes the dynamics) and the anticommutator can be
written in terms of the functions of the type defined in Eq.~\eqref{eq:2}. As
an example, we can write the case of the commutator, which will be used later:
Given two operators $A$ and $B$, with the corresponding functions $f_A$ and
$f_B$, the function associated to the commutator $i[A,B]=i(AB-BA)$
(the imaginary unit is used to preserve hermiticity) is written
as
\begin{equation}
\label{eq:4}
f_{i[A,B]}=\{ f_A, f_B\}=\frac 12 \sum_k \left (\frac {\partial f_A}{\partial
      q^k}\frac{\partial f_B}{\partial p_k}-
\frac {\partial f_A}{\partial
      p_k}\frac{\partial f_B}{\partial q^k} \right ).
\end{equation}
Thus, from the formal point of view, the operation is completely
analogous to the Poisson bracket used in classical mechanics.

Another important property in the set of operators of QM is the corresponding
spectral theory. In any quantum system, it is of the utmost importance to be
able to find eigenvalues and eigenvectors. We can summarize these properties
in the following result: If $f_{A}$ is the function associated to the
observable $A$, then, as a consequence of Ritz's theorem, \cite{Cohen1977}
\begin{itemize}
\item the eigenvectors of the operator $A$ coincide with the critical
      points of the function $f_{A}$, i.e.,
$$
df_{A}(\psi)=0 \Leftrightarrow \psi \text{ is an eigenvector of } A.
$$
 
\item the eigenvalue of $A$ at the eigenvector $\psi$ is the value that
      the function $f_{A}$ takes at the critical point $\psi$.
\end{itemize}

As usual, the dynamics can be implemented in essentially two different forms
(but always in a way which is compatible with the geometric structures
introduced so far): the so-called Schr\"odinger and Heisenberg pictures
\cite{Alonso2011JPA}. In the Heisenberg picture, which is the one we will use
in what follows, the dynamics is introduced by translating the well-known
Heisenberg equation into the language of functions:
\begin{equation}
\label{eq:3}
i\hbar \dot f_A=\{f_A, f_H\},
\end{equation}
being $f_H$ the function associated to the Hamiltonian operator and $A$ any
observable.

\subsection{Geometric quantum statistical mechanics}
\label{sec:gqsm}

\subsubsection{The probability density and the density matrix}

A classical result in QM states that, given a quantum system, the average
value of any observable $A$ can always be computed as the trace of the
observable and some density state $\rho$, as defined in Section
\ref{sec:noti-purity-quant}:
\begin{equation}
\label{eq:1}
\langle A\rangle=\mathrm{Tr}(\rho A).
\end{equation}
This result is known as Gleason theorem (see Ref. \onlinecite{Gleason1957} for
details).

Instead of using the density matrix $\rho$, we can use an alternative approach
which is formally closer to the description of classical statistical systems
and is used, for instance, in Ref.~\onlinecite{Breuer2002}. Consider a probability
distribution $F_Q$ on $M_Q$ and a volume element $d\mu_Q$, satisfying the
properties:
\begin{itemize}
\item $\int_{M_Q}d\mu_Q(\psi) F_Q(\psi)=1.$
\item Expected values can be computed as
\begin{equation}
\label{eq:exp_q_general}
\langle A\rangle=\int_{M_Q}d\mu_Q(\psi) F_Q(\psi)
 \frac{ f_A(\psi)}{\langle\psi|\psi\rangle},
\end{equation}
  for all $f_A$ of the form \eqref{eq:2}; $A$ being a Hermitian operator.
  Notice that we have chosen to integrate over all the states in $M_Q$
  and divide by the norm of the state, as it is done in the final
  section of Ref.~\onlinecite{Alonso2011JPA}. This is equivalent to integrate
  over the states of norm one as it was also done in the first
  sections of Ref.~\onlinecite{Alonso2011JPA}.
\end{itemize}

The canonical symplectic form of $M_Q$ described in Ref.~\onlinecite{Alonso2011JPA}
provides a natural candidate for the volume form since it is also preserved by
the quantum evolution (see Ref.~\onlinecite{Alonso2011JPA} for the technical
details).

Some simple examples for the distribution $F_Q$ can also be provided:

\begin{example}
For the case of the pure state $\rho=\frac{|\psi_0\rangle\langle\psi_0|}{\langle\psi_0|\psi_0\rangle}$,
we can use
\begin{equation}
\label{eq:10}
F_Q(\psi)=\delta(\psi-\psi_0),
\end{equation}
to satisfy the above two equalities. Analogously, a mixed state
$\rho=\sum_{k}p_k\frac{|\psi_k\rangle\langle\psi_k|}{\langle\psi_k|\psi_k\rangle}$ (where $\sum_k p_k=1$ and
$p_k\geq 0$) can be represented by
\begin{equation}
\label{eq:10b}
F_Q(\psi)=\sum_k p_k\delta(\psi-\psi_k).
\end{equation}
In particular, it is straighforward to prove that, in this case,
\begin{align*}
\langle A\rangle &= \int_{M_Q}d\mu_Q(\psi)\sum_k p_k\delta(\psi-\psi_k)
\frac{f_A(\psi)}{\langle\psi|\psi\rangle}  \\
                 &=\sum_k p_k
\frac{f_A(\psi_k)}{\langle\psi_k|\psi_k\rangle}=\sum_k p_k
\mathrm{Tr}( \rho_k 
  A)=\mathrm{Tr}(\rho A),
\end{align*}
for $\rho_k=\frac {|\psi_k\rangle\langle\psi_k|}{\langle\psi_k|\psi_k\rangle}$ and $\rho=\sum_k p_k \rho_k$.
\end{example}

The definition of the function $F_Q$ contains a number of ambiguities which
are explained in detail in Ref.~\onlinecite{Alonso2011JPA}. Essentially, we can add
to any $F_Q$ a term which integrates to zero and has vanishing second-order
momenta. Due to the structure of the observable functions $f_A$, this
modification will not change any average value computed as in
Eq.~(\ref{eq:exp_q_general}), nor will it change the normalization condition
for $F_Q$. This defines an equivalence class of distributions that produce the
same average values, and (through the relationship between distributions and
density matrices) Gleason theorem implies that there is always a distribution
in the class in which the fact that we are dealing with a convex combination
of rank-1 projectors is visible. This is what we used in the example above.

In order to advance in the formulation and illustrate these facts more
precisely, we can consider, for every $|\psi\rangle\in \mathcal{H}$, the
following function:
\begin{equation}
\label{eq:6}
 f_{|\psi\rangle\langle \psi|}(\eta)=\frac 12 \langle\eta 
    |\psi\rangle\langle\psi| \eta\rangle.
\end{equation}

Now, it is easy to see that the following (averaged, now $\psi$-independent)
function
\begin{equation}
\label{eq:11}
f_\rho (\eta) =
 \int_{M_Q}d\mu_Q(\psi)F_Q (\psi)
 \frac{f_{|\psi\rangle\langle\psi|}(\eta)}{\langle\psi|\psi\rangle}
\end{equation}
is the phase-space function associated to a density operator $\rho$
defined by
\begin{equation}
\label{eq:14}  
\rho = 
\int_{M_Q}d\mu_Q(\psi)F_Q(\psi)
\frac{|\psi\rangle\langle\psi|}{\langle\psi|\psi\rangle},
\end{equation}
i.e.,
\begin{equation}
\label{eq:11bis}
f_\rho(\eta)=\frac 12\langle \eta |\rho\eta\rangle.	
\end{equation}

With these definitions, we can now prove the following result:
\begin{theorem}
\label{sec:density}
Let us consider a quantum system in a state described by a probability
distribution $F_Q$. Then, the equivalence class of distributions that produce
the same expected values as $F_Q$ contains as an element:
\begin{equation}
\label{eq:27c}
F_Q\sim \tilde F_Q = \sum_k \lambda_k \delta(\psi-\psi_k),
\end{equation}
where $\{\psi_k\}$ is the set of critical points of $f_\rho$
(as defined in Eqs.~(\ref{eq:11}) and~(\ref{eq:11bis})) and 
$$
\lambda_k=f_\rho(\psi_k).
$$

\end{theorem}
\begin{proof}
It is immediate if we realize that
\begin{align}
\label{eq:28}
\langle A\rangle &=\int_{M_Q}d\mu_Q(\psi) F_Q(\psi)
\frac{f_A(\psi)}{\langle\psi|\psi\rangle} \nonumber \\
&=\mathrm{Tr}\left ( \left (\int_{M_Q}d\mu_Q(\psi) F_Q(\psi)
  \frac{|\psi\rangle\langle\psi|}{\langle\psi|\psi\rangle} \right ) 
  A\right )=\mathrm{Tr}(\rho A).
\end{align}

Indeed, the operator $\rho$ appearing in this expression and defined in
Eq.~(\ref{eq:14}) can be shown to be a density matrix (i.e., $\rho^2 = \rho$,
$\rho^+ = \rho$, and $\mathrm{Tr}\rho = 1$), and Gleason's theorem
guarantees it is unique.

We also know that, if we use the spectral decomposition of $\rho$, i.e.,
$\rho=\sum_k \lambda_k \rho_k$, with
$\rho_k=\frac {|\psi_k\rangle\langle\psi_k|}{\langle\psi_k|\psi_k\rangle}$,
being $\psi_k$ and $\lambda_k$ its eigenvectors and eigenvalues, respectively,
we also have that
\begin{align}
\langle A \rangle  &=  \int_{M_Q} d\mu_Q(\psi) \tilde F_Q(\psi)
\frac{f_A(\psi)}{\langle\psi|\psi\rangle} 
  =  \sum_k \lambda_k
\frac{f_A(\psi_k)}{\langle\psi_k|\psi_k\rangle} \nonumber \\&=
\sum_k \lambda_k \mathrm{Tr}(\rho_k A)=\mathrm{Tr}(\rho A),
\end{align}
as we set out to prove.
\end{proof}

Hence, from Gleason theorem, we know that, among all the equivalent
distributions, there is always one $F_Q$ equal to a convex combination of
Dirac-delta functions. Notice that the function $F_Q$ provides us with all the
information encoded in the density matrix $\rho$. As a probability density, it
allows us to define the average values of the observables, and in the form
$f_\rho$, it allows us to read the spectrum of $\rho$ from the set of critical
points.

This result allows us to realize in terms of $F_Q$ any quantum system: as the
average values coincide with those obtained from the spectral decomposition of
the density matrix, we can use it to implement any desired model. We will see
a practical example in the following sections.

\subsubsection{Geometrical computation of purity}

Finally, we would like to analyze purity in this geometrical context. We saw
in Section \ref{sec:notion-decoherence} that purity preservation is encoded in
the behavior of $\rho$ as a projector. If the evolution of the system
preserves the purity of the density matrix we say that the evolution is
\textbf{purity-preserving}, while in the other case, we call it \textbf{purity
non-preserving}.

Our first task is to express in this geometrical language the concept of
purity. Consider the following expression for any $f_\rho(\psi)$:
\begin{equation}
\label{eq:9}
\langle\rho\rangle:=\int_{M_Q} d\mu_Q(\psi)F_Q(\psi)
  \frac{f_\rho(\psi)}{\langle\psi|\psi\rangle}.
\end{equation}
Then, $\langle\rho\rangle=1$ if the state is pure and $\langle\rho\rangle<1$
if the state is mixed. This is so because, trivially,
$$
\langle \rho\rangle=\mathrm{Tr}(\rho.\rho)=\mathrm{Tr} \rho^2.
$$

Let us now mention how this change in the purity of the state can be detected
in the measurement of average values of observables. Recall that a change in
the purity as above produces a transformation at the level of the states such
that a pure state corresponding to a distribution of the form
$F_Q(\psi)=\delta(\psi-\psi_0)$ becomes a distribution of the form
$F_Q(\psi)=\sum_j p_k\delta(\psi-\psi_k)$, where there is more than one value
of $p_j$ different from zero. Then, it is immediate to prove that the average
value of a generic observable $A$ will be different between one case and the
other.

\section{Purity change in Ehrenfest statistics}
\label{sec:purity-non-pres}

\subsection{The definitions}
\label{s:4.A}
In this section, we will now extend the previous construction to the Ehrenfest
case, by combining it with the approach introduced in
Ref.~\onlinecite{Alonso2011JPA}.

First, let the physical states of our Ehrenfest system correspond to the
points in the Cartesian product
$$
M=M_C\times M_Q,
$$
where $M_C$ is the phase space of the classical system. The physical 
observables will be now functions defined on that manifold. To define  
statistical averages of observables depending on classical and quantum 
degrees of freedom (i.e., functions as $f_A(\xi,\psi)$) we consider
\begin{equation}
\label{eq:15}
\langle A\rangle=\int_{M_C\times
      M_Q}\hspace{-0.2cm}d\mu_{QC}(\xi,\psi) 
  F_{QC}(\xi,\psi)\frac{f_A(\xi,\psi)}{\langle\psi|\psi\rangle},
\end{equation}
where $\xi=(R,P)\in M_C$ represents the classical degrees of freedom,
$\psi=\psi(q,p)\in M_Q$ the quantum ones, and $d\mu_{QC}=d\mu_Qd\mu_C$ is the
volume on the state space manifold $M$.

We can now ask what properties we must require from $F_{QC}$ in order
for Eq.~\eqref{eq:15} to correctly define the statistical mechanics for the
Ehrenfest dynamics (ED). Analogously to what happened in the quantum case,
the conditions are as follows:
\begin{itemize}
    \item The expected value of any constant observable should be equal to that   constant, which implies that the integral on the whole set of states is equal to one:
    \begin{equation}
    \label{eq:16b}
    \int_{M_C\times M_Q}d\mu_{QC}(\xi,\psi)F_{QC}(\xi,\psi)=1.
    \end{equation}
    
    \item The average, for any purely quantum observable $f_A$ of the form in~\eqref{eq:2}, associated to a positive definite Hermitian operator $A$, should be positive. This implies the usual requirement of positive 
    probability density in standard classical statistical mechanics.
\end{itemize}

In Ref.~\onlinecite{Alonso2011JPA}, it was proved that the ED defined on
the manifold $M$ is Hamiltonian with respect to the Poisson bracket
\begin{equation}
  \label{eq:36}
\{\cdot, \cdot \}_{QC}=\{\cdot, \cdot \}_{C}+\frac i\hbar \{\cdot,
\cdot \}_{Q},
  \end{equation}
where $\{\cdot, \cdot \}_{C}$ represents the usual Poisson bracket of
the classical degrees of freedom, and $\{\cdot, \cdot \}_{Q}$
represents the Poisson bracket defined by Eq.~\eqref{eq:4}. 

Being Hamiltonian, we know that we can define an invariant measure on the
space of states $M$. We shall denote such a measure by $d\mu_{QC}$. Thus, the
dynamics defined on the microstates is straightforwardly translated into
the probability density $F_{QC}$ as a Liouville equation:
\begin{equation}
\label{eq:12}
\dot F_{QC}=\{ f_H,F_{QC}\}_{QC},
\end{equation}
where $f_H$ is the Hamiltonian function of the Ehrenfest system:
\begin{equation}
f_H(R,P;q,p):=\sum_{J}\frac{\vec{P_J}^2}{2M_J}+
    \frac{\langle\psi(q,p) |\hat H_e(R)|\psi(q,p)\rangle}
      {\langle\psi(q,p)|\psi(q,p)\rangle}.
\label{eq:26}
\end{equation}
Analogously, the evolution of any function
$f_A(\xi,\psi)$ is given by
\begin{equation}
  \label{eq:39}
  \dot f_A(\xi, \psi)=\{ f_A(\xi, \psi), f_H(\xi, \psi)\}_{QC}.
\end{equation}

Again, this property provides us with a natural candidate for the volume
element $d\mu_{QC}$ (and, for analogous reasons, also for $d\mu_C$) arising
from the symplectic form which gives its Hamiltonian structure to the
Liouville equation in this context. As it happens in the pure quantum case,
this volume form is preserved by the dynamics (see Ref.~\onlinecite{Alonso2011JPA}).

With this in mind, we can consider the analogue of the objects introduced in
the previous section. Hence, given an Ehrenfest system in a state described by
a probability density $F_{QC}(\xi, \eta)$, where $\xi=(R,P)$ and
$\eta=\eta(q,p)$, we can consider the definition of the operator
\begin{equation}
\label{eq:25}
  \rho(\xi):=\int_{M_Q}d\mu_Q(\psi) F_{QC}(\xi,\psi) \frac{|\psi\rangle\langle\psi|}{\langle\psi|\psi\rangle},
\end{equation}
which still depends on the classical variables $\xi$ and, therefore, it can
be interpreted as having turned the phase-space representation of the quantum
part into the more familiar one based on density matrices. Also, since we have
not integrated over $\xi$, this object still represents in a certain way a probability density in the classical part of the space.

As for any $\xi$-dependent operator (see Ref.~\onlinecite{Alonso2011JPA}), we can
define the associated phase-space function:
\begin{equation}
\label{eq:16}
f_\rho(\xi,\eta):=\int_{M_Q}d\mu_Q(\psi) F_{QC}(\xi,\psi)
    \frac{\langle\eta|\psi\rangle\langle\psi |\eta\rangle}
         {\langle\psi|\psi\rangle}.
\end{equation}

It is also possible to integrate again the object in Eq.~(\ref{eq:25}), and 
define:
\begin{equation}
\label{eq:13}
  \rho=\int_{M_C}d\mu(\xi) \rho(\xi),
\end{equation}
which is a purely quantum object encoding the averaged information of the
complete system. Notice that, as it is usually done and in order to lighten
the notation, we will often use the same symbol for different objects (as in
$\rho(\xi)$ and $\rho$), understanding that it is the explicit indication of
the variables on which they depend what distinguishes them notationally. Also,
for simplicity, we use the same symbols for operators in the full-quantum case
in sec.~\ref{sec:gqsm}, and for the ones in the quantum-classical scheme in
this section.

We can now formulate the dynamics in terms of this operator. After a
brief computation, we obtain:
\begin{equation}
\label{eq:24}
\frac d{dt}\rho=i\hbar^{-1} \int_{M_Q\times M_C}\!\!\!\!\!\!\!\!\!\!
  d\mu_{QC}(\xi,\psi)F_{QC}(\xi,\psi)
  \left [H_e(\xi),
    \frac{|\psi\rangle\langle\psi|}{\langle\psi|\psi\rangle}\right ],
\end{equation}
where $H_e(\xi)$ is the very electronic Hamiltonian defined at the beginning
of the paper. In ED, considered statistically, this equation represents the
analogue for the mixed case of von Neumann's equation. 

\begin{example}
If we consider a single (`pure') state of the classical system $\xi_0$ and a 
pure state of the quantum system $|\psi_0\rangle$, i.e., 
$$
F_{QC}(\xi,\psi)=\delta (\xi-\xi_0)\delta(\psi-\psi_0),
$$
we obtain
$$
f_\rho(\xi, \eta)=\delta (\xi-\xi_0)\frac{|\langle \psi_0, 
        \eta\rangle|^2}{\langle\psi_0|\psi_0\rangle}.
$$
Analogously, the associated operator in Eq.~\eqref{eq:25} reads
$$
\rho(\xi)=\delta(\xi-\xi_0)\frac{|\psi_0\rangle\langle\psi_0|}
              {\langle\psi_0|\psi_0\rangle},
$$
and the one in Eq.~\eqref{eq:13} is
$$
\rho=\frac{|\psi_0\rangle\langle\psi_0|}
              {\langle\psi_0|\psi_0\rangle},
$$
as expected.
\end{example}

We can also consider the two marginal distributions:
\begin{itemize}
\item The distribution in $M_Q$ obtained by integrating out the classical 
 degrees of freedom:
\begin{equation}
\label{eq:14b}
F_Q(\psi)=\int_{M_C} d\mu_C (\xi)F_{QC}(\xi,\psi),
\end{equation}
\item and the corresponding classical version in $M_C$:
\begin{equation}
\label{eq:17}
F_C(\xi)=\int_{M_Q}d\mu_Q(\psi)F_{QC}(\xi, \psi).
\end{equation}
\end{itemize}
Obviously both functions are distribution functions on the corresponding
manifolds with analogous properties to $F_{QC}$, and they could be used to
compute expected values of functions depending only on $\psi$ or $\xi$,
respectively.

\begin{example}
It is immediate to check that the definitions make sense for distributions of 
the form
\begin{equation}
\label{eq:22}
F_{QC}(\xi, \psi)=\delta(\xi-\xi_0)\sum_kp_k\delta(\psi-\psi_k),
\end{equation}
i.e., for `pure' classical part and a quantum-mixed one canonically expressed
with deltas.

In this case, the marginal distributions are of the form
\begin{equation}
\label{eq:23}
F_Q(\psi)=\sum_kp_k\delta(\psi-\psi_k), \quad F_C(\xi)=\delta(\xi-\xi_0).
\end{equation}
\end{example}

Also notice that, in terms of the quantum marginal distribution $F_Q$, we can
write the density matrix in Eq.~\eqref{eq:13} as
\begin{equation}
\label{eq:35}
\rho=\int_{M_Q}d\mu_Q(\psi)F_Q(\psi)\frac{|\psi\rangle\langle\psi|}
    {\langle\psi|\psi\rangle}.
\end{equation}

Once we have recovered the needed ingredients, we can discuss the quantum
purity of a system governed by ED:

\begin{definition}
We say that a quantum-classical system is \textbf{quantum-pure} if
and only if 
\begin{equation}
\label{eq:21}
\mathrm{Tr} \rho^2 = \langle \rho \rangle=\int_{M_Q}d\mu_Q(\psi)
  F_Q (\psi)\frac{f_{\rho}(\psi)}{\langle\psi|\psi\rangle}=1,
\end{equation}
being $\rho$ the one defined in eqs.~\eqref{eq:13} and~\eqref{eq:35}.

In case that the state of a system does not satisfy the condition above, we
say that it is \textbf{quantum-mixed}.
\end{definition}

For the sake of completeness, and in order to better connect with the purely
quantum case we discussed in sec.~\ref{sec:gqsm}, we can consider now the
function obtained from $f_\rho$ in Eq.~\eqref{eq:16} averaging directly over
$M_C$, i.e.,
\begin{align}
\label{eq:18}
f_\rho(\eta) & = \int_{M_C}d\mu_C(\xi)f_\rho(\xi, \eta)  \nonumber\\
 & =\int_{M_C\times M_Q}
   d\mu_{QC}(\xi,\psi)F_{QC}(\xi,\psi)\frac{\langle\eta|\psi\rangle
   \langle\psi |\eta\rangle}{\langle\psi|\psi\rangle} \nonumber \\
&= \int_{M_Q} d\mu_Q(\psi)F_Q(\psi)\frac{\langle\eta|\psi\rangle \langle\psi
    |\eta\rangle}{\langle\psi|\psi\rangle}.
\end{align}

This function plays the role of the function in~\eqref{eq:11} in the pure
case. Indeed, we have that
\begin{align}
\label{eq:19}
f_\rho(\eta)&=\langle \eta |  
\left ( \int_{M_C\times M_Q}d\mu_{QC}F_{QC}(\xi,\psi)\frac{|\psi
\rangle\langle\psi |}{\langle\psi|\psi\rangle} \right ) | \eta\rangle  
 \nonumber \\
 &=\langle \eta |  \left ( \int_{M_Q}d\mu_{Q}F_{Q}(\psi)\frac{|\psi
\rangle\langle\psi |}{\langle\psi|\psi\rangle} \right ) | \eta\rangle\ 
 = \langle \eta | \rho \eta \rangle,
\end{align}
and hence it corresponds to the quantum expected value of the operator $\rho$
in eqs.~\eqref{eq:13} and~\eqref{eq:35}, whose full-quantum analogue is
the one in~\eqref{eq:14} in sec.~\ref{sec:gqsm}.

\subsection{The application: transferring uncertainty between the classical
and quantum parts}
\label{sec:transferring}

Consider the following initial distribution evolving under ED:
$$
F_{QC}(0)=\delta(\xi-\xi_0)\delta(\psi-\psi_0).
$$

This system is completely deterministic and, therefore, the Liouville
equation will produce exactly, as a solution, the integral curves of ED
$(\xi(t), \psi(t))$. Thus we can write:
$$
F_{QC}(t)=F_{QC}(\xi(t), \psi(t))=\delta(\xi-\xi(t))\delta(\psi-\psi(t)).
$$

Consider now a slightly more complex system, consituted by a distribution
of $N$ equally probable classical states, and a pure quantum state
at $t = 0$:
\begin{equation}
\label{eq:27}
F_{QC}(0)=\left (\frac 1N \sum_{k=1}^N\delta(\xi-\xi_0^k)\right ) 
  \delta(\psi-\psi_0).
\end{equation}

We can also write the marginal distributions as we did in Section \ref{s:4.A}:
\begin{equation}
\label{eq:29c}
F_C(0)=\frac 1N \sum_{k=1}^N\delta(\xi-\xi_0^k)\ ; \quad F_Q(0)=\delta(\psi-\psi_0).
\end{equation}

The evolution of such a system becomes
\begin{equation}
\label{eq:27b}
F_{QC}(t) =\frac 1N \sum_{k=1}^N
    \delta(\xi-\Phi_\xi^*(\xi_0^k,\psi_0;t))
    \delta(\psi-\Phi_\psi^*(\xi_0^k,\psi_0;t)), 
\end{equation}
where $(\Phi_\xi^*(\xi_0^k,\psi_0;t),\Phi_\psi^*(\xi_0^k,\psi_0;t))$
represents the Ehrenfest trajectory having $(\xi_0^k, \psi_0)$ as initial
condition.

The evolved marginal distribution in the quantum manifold is now:
\begin{equation}
\label{eq:28b}
F_Q(t)=\frac{1}{N} \sum_{k=1}^N\delta (\psi-\Phi_\psi^*(\xi_0^k,
  \psi_0;t)),
\end{equation}
and then, using Eq.~\eqref{eq:21}, we have:
\begin{align}
\label{eq:30}
f_\rho(t)=&\frac 1N \sum_{k=1}^N
  \frac{|\langle \eta| \Phi_\psi^*(\xi_0^k,\psi_0;t)\rangle |^2}
       {\langle \Phi_\psi^*(\xi_0^k,\psi_0;t) |\Phi_\psi^*(\xi_0^k,
  \psi_0;t)\rangle },
\end{align}
where, of course,
\begin{equation}
\label{eq:30b}
f_\rho(0)=| \langle \eta|\psi_0)\rangle|^2.
\end{equation}

The associated density matrix at time $t$ reads:
\begin{align}
\label{eq:31}
\rho(t)=&\frac 1N \sum_{k=1}^N\frac{|\Phi_\psi^*(\xi_0^k,
  \psi_0;t)\rangle
  \langle\Phi_\psi^*(\xi_0^k, \psi_0;t)|}{\langle \Phi_\psi^*(\xi_0^k,
  \psi_0;t) |\Phi_\psi^*(\xi_0^k,
  \psi_0;t)\rangle },
\end{align}
and at time $t=0$:
\begin{equation}
\label{eq:33}
\rho(0)=\frac{|\psi_0\rangle\langle\psi_0|}{\langle\psi_0|\psi_0\rangle}.
\end{equation}

Now, the purity at time $t=0$ is
\begin{equation}
\label{eq:32}
\langle \rho(0) \rangle=\int_{M_Q}d\mu_QF_Q(0) f_\rho(0)=1,
\end{equation}
but at a general time $t$,
\begin{align}
\label{eq:5}
\langle \rho(t) \rangle &=\int_{M_Q}d\mu_QF_Q(t)  f_\rho(t) \nonumber
\\
&= \frac
1N\sum_{k,j=1}^N\frac{\langle \Phi_\psi^*(\xi_0^k, 
      \psi_0;t)|\Phi_\psi^*(\xi_0^j, \psi_0;t) \rangle^2}{\|
  \Phi_\psi^*(\xi_0^k, \psi_0;t)\|^2\|
  \Phi_\psi^*(\xi_0^j, \psi_0;t)\|^2}.
\end{align}

Hence, the purity of the system seems to evolve in time, in general in an
involved way. We can compute, though, its time derivatives, in order to
get an idea about its initial evolution. After some calculations which we
detail in appendix~\ref{app:derivatives}, one can see that, for the initial
state considered in this section, we have
\begin{equation}
\label{eq:34short}
\frac d{dt} \langle \rho(t) \rangle \Big|_{t=0} = 0,
\end{equation}
and
\begin{align}
\label{eq:40short}
\frac{d^2}{dt^2}\langle\rho(t)\rangle\Big|_{t=0}  =& 
 \hbar^2 \langle \psi_0 |(H_e(\xi_0^1)-H_e(\xi_0^2)) |\psi_0\rangle^2
 \nonumber \\
 &- \hbar^2 \langle \psi_0 |(H_e(\xi_0^1)-H_e(\xi_0^2))^2|\psi_0\rangle,
\end{align}
which is negative definite unless the expectation value of
$H_e(\xi_0^1)-H_e(\xi_0^2)$ at $|\psi_0\rangle$ vanishes. This means that the
purity evolves from its initial value, and it does so by decreasing, which is
entirely expected if we realized that we started from a state in which the
purity is maximal.

Thus, we can see that the evolution has made a quantum-pure system become a
quantum-mixed one, and we can claim that:
\begin{theorem}
The statistical Ehrenfest evolution is  purity non-preserving for
Hamiltonians with quantum-classical couplings.  
\end{theorem}

This is the main result of our paper: even though the ED of a
single-trajectory state does preserve purity, when we consider a statistical
state the behavior changes. In the case above (and as we will see
numerically in the next section), we show how it is possible to transfer
uncertainty from the classical domain into the quantum one. Analogously, it is
straightforward to see that an analogous process happens when we consider a
single classical state coupled to an ensemble of quantum
states.  This uncertainty transfer takes place whenever the
dynamics of the quantum and the classical subsystems are coupled to
each other, since such coupling produces a splitting of the
trajectories from the given initial conditions and thus the mixing of
the final states.

\section{Numerical example on a simple Ehrenfest system}
\label{sec:ejemplo}

To illustrate the concepts introduced in the previous section and to see the
purity change in a complex numerical case, let us consider a situation with an
equally probable initial distribution as in Eq.~\eqref{eq:27} of classical
particles with $N=5$ and a quantum part constituted by a 10-level
system.
Such a
system has been used in the literature as a natural framework for molecular
dynamics (see for example Refs.~\onlinecite{Tully1990b,Tully1998Book,Bastida2007JCP}). We consider
a Hamiltonian function for the quantum-classical system of the following form:
\begin{equation}
\label{eq:41}
f_H=J+\langle\psi| A+\epsilon J\cos \theta B|\psi\rangle,
\end{equation}
where $(\theta, J)$ are canonically conjugated classical variables and $A$ and
$B$ are Hermitian matrices acting on the quantum vector space
$\mathbb{C}^{10}$. The classical part is written in action-angle
coordinates to simplify the analysis.

The dynamics of the system is obtained from the solutions of the $N$ different
trajectories with initial conditions defined by each one of elements of the
initial classical distribution $\{ \theta_1(0), \theta_2(0), \theta_3(0),
\theta_4(0), \theta_5(0)\}$. The resulting distribution takes thus the form
given by Eq.~\eqref{eq:27b}. As it can be found in the Supplementary Material,
for all the trajectories presented below the initial conditions for the
classical subsystem are chosen as
 \begin{align*}
\theta_1(0) &= 0.9766548288669266, \\
\theta_2(0) &= 0.5013694871260747, \\
\theta_3(0)&=  0.9052199783160014, \\
\theta_4(0)&= 0.5068075140327187, \\
\theta_5(0) &= 0.9543157645144570 \ .
\end{align*}
The points were fixed as five random points in the interval $[0,1]$. Choosing
different initial conditions leads to equivalent results.

The quantum initial condition is chosen to be, for all five trajectories,
$$
\psi_0=(1, 0, \ldots, 0).
$$

The evolution defines thus a equiprobable distribution of $N$
quantum-classical single trajectories
$(\Phi_\xi^*(\xi_k,\psi_0;t),\Phi_\psi^*(\xi_k,\psi_0;t))\in M_C\times M_Q$.
For each trajectory, the equations of motion are of Ehrenfest type and they
are given by eqs.~\eqref{eq:ehrenfest1}--\eqref{eq:ehrenfest3}, with $H_e = 
A+\epsilon J \cos \theta B$.

These dynamical equations exhibit two different regimes:
\begin{itemize}
\item For $\epsilon=0$ or $J(0)=0$, the classical and the quantum subsystems
  evolve uncoupled. The purity of the quantum subsystem is always
  equal to one.
\item For non-zero coupling constants and classical momenta, the
  behavior of the system depends sharply on the initial conditions:
  the evolution from different initial conditions for the classical subsystems
  is very different. The purity tends then to evolve to the purity
  corresponding to a set of five random projectors on 1-dimensional
  subspaces of the Hilbert space.
\end{itemize}

The evolution of the purity of the resulting system is obtained from
Eq.~\eqref{eq:5} after integrating the dynamics numerically. Some interesting
observations can be extracted from the results:

\begin{itemize}
\item In fig.~\ref{fig:3} we represent the evolution of the purity for a fixed
value of the coupling constant $\epsilon$ and increasing value for the initial
condition of the classical momentum $J_0$. We see how this makes the system
change its originally integrable behavior and become more and more chaotic.

\begin{figure}[!ht]
\centering
\includegraphics[width=8cm]{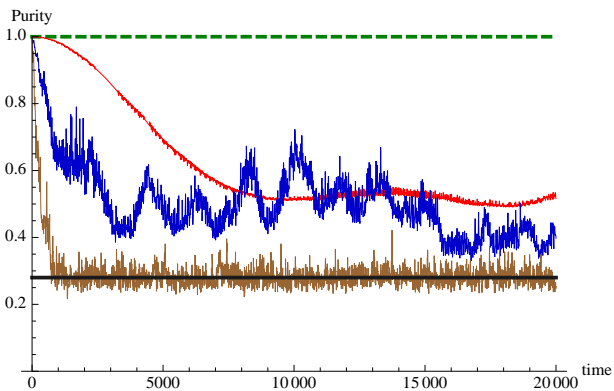}
\caption{Evolution of the purity for $\epsilon=0.1$ and $J_0=0$ (dashed green
line), $J_0=0.3$ (red line), $J_0=1$ (blue line) and $J_0=1.5$ (brown
line).We also depict the reference (black line) of the level of purity of a
distribution of $N=5$ random projectors on $\mathbb{C}^{10}$.}
\label{fig:3}
\end{figure}

\item Instead, we can consider a fixed value of the initial classical momentum
and increase the value of the coupling.  It can be remarked that the system
reaches the level of purity of the set of random projectors much faster than in the previous case, (see fig.~\ref{fig:4}).

\begin{figure}[!ht]
\centering
\includegraphics[width=8cm]{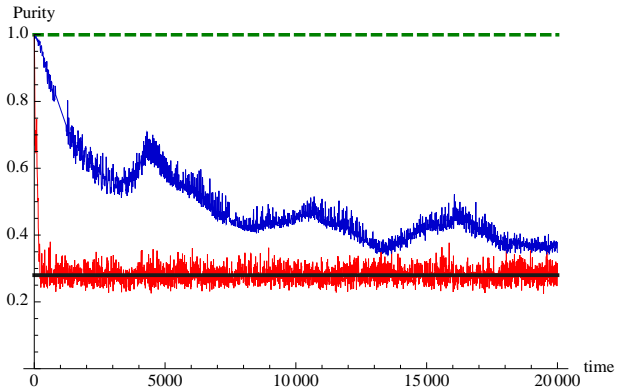}
\caption{Evolution of the purity for $J_0=0.8$ and $\epsilon=0$ (dashed green
line), $\epsilon=0.1$ (blue line), and $\epsilon=0.2$ (red line).  Again, the reference (black line)
represents the level of purity of a distribution of $N=5$ random projectors on
$\mathbb{C}^{10}$.}
\label{fig:4}
\end{figure}
 
\end{itemize}

The interesting behaviour of the purity shown for some values of the
parameters in figs.~\ref{fig:3} (brown line) and~\ref{fig:4} (red line), in
which its value rapidly decreases to a given low one and fluctuates around it
after that, can be explained as a consequence of the  dynamics
mentioned before. If the system exhibits sensitive dependence on the initial
conditions, after a small lapse of time, the different trajectories become
decorrelated from one another in the Hilbert space. Therefore, the density
matrix becomes the normalized sum of $N$ rank one projectors chosen at random.
In this situation, it can be shown (see appendix \ref{app:random}) that the purity is
distributed around an expected value
\begin{equation}
\label{eq:43}
\operatorname{E}[\mathrm{Tr}(\rho^2)] = \frac{N+D-1}{ND},
\end{equation}
which is represented by the black straight lines in the figures, with 
fluctuations of size
\begin{align}
\label{eq:fluctuations}
\sigma  := &\left( \operatorname{E}[\mathrm{Tr}(\rho^2)^2]
 - \operatorname{E}[\mathrm{Tr}(\rho^2)]^2\right)^{1/2}  \nonumber \\
 = &
\frac{\sqrt{2}}{ND}\left(\frac{(N-1)(D-1)}{N(D+1)}\right)^{1/2},
\end{align}
where $N$ is the total number of trajectories used and $D$ is the dimension of
the quantum Hilbert space. Notice that if $D$ is very large, the expected
value for the purity tends to its minimal value $1/N$ and the fluctuations
tend to zero. Remember that the degree of mixture of a quantum system
ranges from the pure state case (i.e., purity equals to one), for
which the density operator is a projector on a one dimensional
subspace of the Hilbert space and the maximal mixture case (thus minimum
purity) which   corresponds to a density operator which is
proportional to the identity matrix. As it must have trace equal to
one, the proportionality factor is equal to the inverse of the
dimension of the Hilbert space.

Figs.~\ref{fig:3} and~\ref{fig:4} show also the effect of the coupling between
the classical and the quantum subsystems, as well as the effect of the
momentum (or equivalently, the energy) of the classical particles on the
evolution of purity. The greater the strength of the coupling and the energy
of the classical particle, the faster the system reaches its asymptotic
behavior. Although clearly the coupling has a stronger effect. Thus,
analogously to what happens in the case of molecular dynamics, where the
velocity of the classical nuclei induces the coupling between all the
eigenstates of the electronic Hamiltonian, we see here that it also has the
effect of mixing the quantum part if the Ehrenfest system is treated
statistically.

From our analysis and example, we can then conclude that statistical Ehrenfest
dynamics provides a framework in which the evolution affects the quantum
dynamics in at least one way decoherence does. Changing the degree of mixture
of the quantum state is certainly one of the most relevant effects of
electronic decoherence on molecular systems, although further work is required
to analyze whether or not other decoherence effects can be explained by our
construction.

\section{Conclusions and future work}
\label{sec:concl-future-work}

An appropriate description of the electronic decoherence in molecular systems
is on the wishlist of every quantum-classical dynamics scheme. In which amount
each theoretical model includes the sought effects is a complicated question
whose answer will depend both on the model and on the intended application.
Loosely speaking, we could expect the different models to range from `no
electronic decoherence at all' (e.g., BOMD) to `a perfect description of
quantum electronic decoherence' (say, full quantum dynamics of electrons and
nuclei), with most of them lying somewhere in between the two extremes.
 Several methods \cite{Truhlar2007Book, Zhu2005JCTC,
Prezhdo1999JCP, Subotnik2010JCP, Horsfield2006RPP, Tully1990JCP, Tully1990b,
Tully1998Book, Landry2011JCP,Bedard2005JCP,SunWangMiller:1998,SubShen:2011}
have been developed to deal with the problem but, to our knowledge, without a
definitive solution.

Pure Ehrenfest dynamics (ED for single states) clearly lies in the `no
decoherence at all' side of the spectrum, since it does not even allow the
change of purity. However, if we are willing to accept the possibility that
the electrons evolve into a mixed state, i.e., a statistically uncertain
electronic state, then, to perform a coherent analysis, we should also allow
the initial conditions of the nuclei to be described statistically. This point
of view was used in some of the papers mentioned above when considered within
surface-hopping or decay-of-mixing formalisms. In this work we implement the
idea using the statistical description of Ehrenfest dynamics introduced in
Ref.~\onlinecite{Alonso2011JPA} combined with a specially convenient geometric
formalism we introduce. Using these methods, one can observe that a system
starting from uncertain nuclear initial conditions and a pure electronic state
evolves into a situation in which the electronic state is no longer pure but
mixed, i.e., ED is capable of transferring uncertainty between the nuclei and
the electrons. Besides, our method provides us also with tools to compute in a
simple but rigorous way the purity change, and it indicates an interesting
dependence of the effect on the number of trajectories $N$ and the dimension
of the Hilbert space $D$, as well as on the coupling and the velocity (or
energy) of the classical system. `How much' decoherence the statistical ED
contains, as related to interesting practical applications or to other mixed
quantum-classical schemes,\cite{Bittner1997,Bittner1995} is a complex and
important question that we shall explore in future works.

\section*{Acknowledgements}

We thank A. Castro for many fruitful discussions on these topics and
many interesting suggestions on previous versions of the manuscript.

This work has been supported by the research grants E24/1, 24/2 and E24/3 (DGA, Spain),
FIS2009-13364-C02-01,  FPA2009-09638 and Fis2009-12648 (MICINN,
Spain), and ARAID and Ibercaja grant for young researchers (Spain).

\appendix

\section{Purity of a sum of random projectors of rank one}
\label{app:random}

In this appendix we derive the expected value of the purity and its
fluctuations when the density matrix is obtained from the sum of decorrelated,
random, rank-one projectors.

The purity for a density matrix of the form
$$
\rho=\frac1N\sum_{j=1}^N\frac{|\psi_j\rangle\langle\psi_j|}
{\langle\psi_j|\psi_j\rangle}
$$
is
$$
{\rm Tr}(\rho^2)=\frac1{N^2}\left ( N+2\sum_{j<k} \chi_{jk}\right),
$$
where we have defined
$$\chi_{jk}:=\frac{|\langle\psi_j|\psi_k\rangle|^2}
{\langle\psi_j|\psi_j\rangle\langle\psi_k|\psi_k\rangle}.
$$

Next, we will determine the probability density for $\chi_{12}$ when $\psi_1$
and $\psi_2$ are two decorrelated random vectors.

Given the global $U(N)$ symmetry of the problem we can take the first vector 
to be
$$
\psi_1=(a,0,\cdots, 0),\quad a\in\R,
$$
and the second chosen at random. If we denote
$$
\psi_2=(q_1+ip_1,q_2+ip_2,\dots,q_N+ip_N),
$$
we get 
$$
\chi:=\chi_{12}=\frac{(q_1^2+p_1^2)}{r^2},
$$
where $r^2:=\sum_j (q_j^2+p_j^2)$.

We also need the adequate probability measure in $M_Q$ that distributes the
random vector $\psi_2$. It can be defined by
$$
d\pi_Q=f(r^2)d\mu_Q,
$$
where $f$ is a positive function chosen so that 
$$
\int_{M_Q}d\pi _Q=1.
$$
As we shall see, the actual form of $f$ is not relevant
for the distribution of the purity.

It will be convenient to write the probability measure in 
$M_Q=\R^2\times\R^{2D-2}$ in the following way:
$$
d\pi_Q=f(z^2+R^2) z dz d\theta R^{2D-3}dR d\Omega_{2D-3},
$$
where $(z,\theta)$ represents polar coordinates in the plane $(q_1,p_1)$, $R$
is the radial coordinate in $\R^{2D-2}$ (i.e., $R^2:=\sum_{j>1}
(q_j^2+p_j^2)$), while the volume element $d\Omega_{2D-3}$ stands for the
angular coordinates in $\R^{2D-2}$. In these coordinates $r=\sqrt{z^2+R^2}$
and $\chi=z^2/r^2$.

The next step is to perform the change of variables from $(R,z)$ to
$(r,\chi)$. Taking into account that the Jacobian is
$$
J=\frac r {2\sqrt{\chi(1-\chi)}},
$$
one obtains
$$
d\pi_Q= \frac12 f(r^2)(1-\chi)^{D-2}d\chi 
r^{2D-1}dr d\theta d\Omega_{2D-3},
$$
and marginalizing out all variables except $\chi$ we find the
needed measure:
$$
d\pi_\chi=(D-1)(1-\chi)^{D-2} d\chi.
$$

Once we have determined the probability distribution for $\chi$ we can compute
the average value for the purity. Using the expression at the beginning of
this appendix,
\begin{equation}
\operatorname{E}[ \mathrm{Tr}(\rho^2)] = 
\frac1{N^2}\left ( N+2\sum_{j<k} \operatorname{E}[\chi_{ij}]\right)
 = 
\frac1{N}\Big( 1+(N-1) \operatorname{E}[\chi]\Big),
\end{equation}
where we have used that all random variables $\chi_{ij}$ are 
identically distributed.

Finally, given that 
$$ 
\operatorname{E}[\chi]=\int_0^1\chi d\pi_\chi=\frac1D,
$$
we obtain the sought result:
$$
\operatorname{E}[ \mathrm{Tr}(\rho^2)]=\frac{N+D-1}{ND}.
$$

As for the fluctuation, one has
\begin{align}
\sigma^2 & = 
 \frac4{N^4} \sum_{j < k} 
 (\operatorname{E}[\chi_{jk}^2] - \operatorname{E}[\chi_{jk}]^2) 
\nonumber \\ 
&= 
 \frac{2(N-1)}{N^3}
 (\operatorname{E}[\chi^2] - \operatorname{E}[\chi]^2) 
=
 \frac{2(N-1)}{N^3}\frac{(D-1)}{D^2(D+1)},
\end{align}
where we have used that
$$
\operatorname{E}[\chi^2]=\frac2{D(D+1)}.
$$

\section{Derivatives of the purity}
\label{app:derivatives}

In sec.~\ref{sec:transferring}, we explained how an initial uncertainty in the
classical part of the state of an Ehrenfest system produces a change in the
purity at $t=0$, even if the initial state is quantum-pure. In this appendix,
we present in more detail the calculations that led us to that conclusion.

First of all, we need the first time-derivative of the purity. We know
that the evolution of the $\rho(t)$ is given by Eq.~\eqref{eq:24}.
Then, the evolution of the purity can be written as:
\begin{align}
\label{eq:34}
\frac d{dt} \langle \rho(t) \rangle =& 
 2\mathrm{Tr}\big( \dot{\rho}(t) \rho(t) \big)\nonumber \\ 
=& 
\int d\mu d\mu' F F'
   \mathrm{Tr}\big( i\hbar^{-1} [ P_\psi, H_e ] \cdot P_{\psi'} \big)
   \nonumber \\
 = & i\hbar^{-1} \int d\mu d\mu' F F'
   \mathrm{Tr}\big( [ P_\psi, P_{\psi'} ] \cdot H_e \big),
\end{align}
where the integral is
taken over $(M_Q \times M_C)^2$, and we have denoted
{\allowdisplaybreaks
\begin{subequations}
\label{eq:abb}
\begin{align}
d\mu & := d\mu_{QC}(\xi,\psi), \label{eq:abb_a} \\
d\mu' & := d\mu_{QC}(\xi',\psi'), \label{eq:abb_b} \\
F & := F_{QC}(\xi,\psi), \label{eq:abb_c} \\
F' & := F_{QC}(\xi',\psi'), \label{eq:abb_d} \\
H_e & := H_e(\xi), \label{eq:abb_e} \\
H'_e & := H_e(\xi'), \label{eq:abb_f} \\
P_\psi & := \frac{|\psi\rangle\langle\psi|}{\langle\psi|\psi\rangle},  \label{eq:abb_g} \\
P_{\psi'} & := \frac{|\psi'\rangle\langle\psi'|}{\langle\psi'|\psi'\rangle}.  \label{eq:abb_h}
\end{align}
In the last step, we also used that
\begin{align*}
\mathrm{Tr} \big([ P_\psi, H_e ] \cdot P_{\psi'}  \big) &=\mathrm{Tr}
\big(  P_\psi\cdot H_e\cdot P_{\psi'} -  H_e\cdot P_\psi\cdot
P_{\psi'} \big)  \\
&=\mathrm{Tr}
\big(   H_e\cdot P_{\psi'}\cdot P_\psi -  H_e\cdot P_\psi\cdot
P_{\psi'} \big)\\
&=\mathrm{Tr}\big( H_e\cdot [ P_\psi, P_{\psi'} ] \big) \\
&=
\mathrm{Tr}\big( [ P_\psi, P_{\psi'} ] \cdot H_e \big).
 \end{align*}

\end{subequations}
}

Also, as it is common in statistical dynamics, we can assign the
time-evolution of the state to the probability distribution $F_{QC}$ and see
the objects $\xi$, $\xi'$ and $|\psi\rangle$, $|\psi'\rangle$ just as the
initial conditions, or we can alternatively think that $F_{QC}$ is the static
distribution of initial conditions and consider that the time-evolving objects
are $\xi$, $\xi'$ and $|\psi\rangle$, $|\psi'\rangle$. Either dynamical image
is valid, and the two of them produce, of course, the same result, but we have
performed the calculation thinking in the second way, which looked to us
slightly more direct.

Now, in our example, $P_\psi(t=0)=P_{\psi'}(t=0)=P_{\psi_0}$. Then, we see that
the commutator $[ P_\psi, P_{\psi'} ]$ vanishes. 
Thus, we can
conclude that:
\begin{equation}
\label{eq:34shortbis}
\frac d{dt} \langle \rho(t) \rangle \Big|_{t=0} = 0.
\end{equation}

Using again Eq.~\eqref{eq:ehrenfest1}-\eqref{eq:ehrenfest3}, we can also compute the second derivative of
the density matrix,
\begin{align}
\label{eq:37}
\ddot{\rho}(t) =
  \int d\mu F &\left ( i\hbar^{-1} 
  \big( \big[ \{ H_e, f_H \}_C, P_\psi \big] \big) \right .  \nonumber \\
  &- \hbar^{-2} \left .\big( \big[ H_e, \big[ H_e, P_\psi ] \big] \big) \right),
\end{align}
where we have used the same notation in eqs.~\eqref{eq:abb}, and the integral
is this time extended to $M_Q \times M_C$. With this expression, we can 
compute the second derivative of the purity:
\begin{equation}
\label{eq:38}
\frac{d^2}{dt^2}\langle\rho(t)\rangle=\mathrm{Tr}\Big(\rho(t)\ddot
  \rho(t)+(\dot \rho(t))^2\Big).
\end{equation}

Now, using eqs.~\eqref{eq:24}, \eqref{eq:34} and~\eqref{eq:37} we can
calculate a more explicit form for this second derivatives in terms of the
objects associated to the geometric formalism
\begin{eqnarray}
\label{eq:399}
\frac{d^2}{dt^2} \langle\rho(t)\rangle = 
\int &d\mu& d\mu' F F' 
\left[ i\hbar^{-1}
 \mathrm{Tr} \Big( \big[ \{ H_e, f_H \} , P_\psi \big] \cdot P_{\psi'}
 \Big)
\right .\nonumber \\
 &-& \hbar ^{-2} \mathrm{Tr} \Big( \big[ H_e, [H_e, P_\psi] \big] 
  \cdot P_{\psi'} \Big)  \nonumber \\
 &-& \left .\hbar ^{-2} \mathrm{Tr} \Big( \big[H_e, P_\psi\big]
  \cdot \big[H'_e, P_{\psi'}\big] \Big) \right] \nonumber \\
  =  
\int &d\mu& d\mu' F F' 
\left[ i\hbar^{-1} 
 \mathrm{Tr} \Big( \big[ \{ H_e , f_H \} , P_\psi \big] \cdot
 P_{\psi'} \Big)
\right .\nonumber \\
& +&  \hbar ^{-2} \mathrm{Tr} \Big( \big[H_e, P_\psi\big]
  \cdot \big[H_e, P_{\psi'}\big] \Big)  \\
 &-& \left .
  \hbar ^{-2} \mathrm{Tr} \Big( \big[H_e, P_\psi\big]
  \cdot \big[H'_e, P_{\psi'}\big] \Big) \right] \nonumber \\
  =  
\int &d\mu& d\mu' F F' 
\left[ i\hbar^{-1} 
\mathrm{Tr} \Big( \big[ \{ H_e , f_H \} , P_\psi \big] \cdot P_{\psi'}
\Big)
\right . \nonumber \\
&+& \left .\hbar ^{-2} \mathrm{Tr} \Big( \big[H_e, P_\psi\big]
 \cdot \big[(H_e - H'_e), P_{\psi'}\big] \Big) \right],\nonumber
\end{eqnarray}
where we used that 
$$
\mathrm{Tr} \Big( \big[ H_e, [H_e, P_\psi] \big] 
  \cdot P_{\psi'} \Big) =
 \mathrm{Tr} \Big( \big[H_e, P_\psi\big]
  \cdot \big[H_e, P_{\psi'}\big] \Big).
$$

Using this expression, it is finally straightforward to compute the second
derivative at $t=0$ for the distribution $F_{QC}$ given by Eq.~\eqref{eq:27b}:
\begin{align}
\label{eq:40}
\frac{d^2}{dt^2}\langle\rho(t)\rangle \Big|_{t=0} =&
 - \frac{\hbar^{-2}}2 \mathrm{Tr} \left( \big[ H_e(\xi_0^1) - H_e(\xi_0^2) ,
    P_{\psi_0} \big]^2 \right) \nonumber \\
 =&-  \hbar^{-2} \langle \psi_0 |\big( H_e(\xi_0^1) - H_e(\xi_0^2) \big)^2
    |\psi_0\rangle \nonumber \\
 &+ \hbar^{-2} \langle \psi_0 |\big( H_e(\xi_0^1) - H_e(\xi_0^2) \big)  
    |\psi_0\rangle^2.
\end{align}


%

\end{document}